\documentclass[12pt,twoside]{article}
\usepackage{amsmath,amsfonts,amssymb,amsthm,bm}
\usepackage{natbib}
\usepackage{graphicx}
\usepackage{color}
\usepackage{subfigure}

\usepackage{setspace}

\theoremstyle{definition}
\newtheorem{defn}{Definition}

\theoremstyle{remark}
\newtheorem*{rk}{Remark}
\theoremstyle{plain}

\newtheorem{claim}[defn]{Claim}
\newtheorem{lemma}[defn]{Lemma}\newcommand{\kth}[1]{${#1}^{\textrm{th}}$}

\begin{document}
\singlespacing
\title{Modelling time to event with observations made at arbitrary times} 
\author{MATTHEW SPERRIN$^\ast$\\[4.1pt]
{\it \small Department of Mathematics and Statistics, Lancaster University, UK}
\\[2pt]
{\small m.sperrin@lancaster.ac.uk}\\[8.1pt]
IAIN BUCHAN\\[4.1pt]
{\it \small North-west Institute for Bio-Health Informatics,}\\ {\it \small School of Community Based Medicine,}\\
{ \it \small University of Manchester, UK}
\\
}
\maketitle

\pagestyle{myheadings}
\headsep=25pt

\markboth{M. Sperrin and I. Buchan}{Modelling time to event with observations made at arbitrary times}

\maketitle
\renewcommand{\thefootnote}{\fnsymbol{footnote}}
\footnotetext{$^\ast$ To whom correspondence should be addressed.}

\doublespacing
\begin{abstract}
We introduce new methods of analysing time to event data via extended versions of the proportional hazards and accelerated failure time (AFT) models. In many time to event studies, the time of first observation is arbitrary, in the sense that no risk modifying event occurs. This is particularly common in epidemiological studies. We show formally that, in these situations, it is not sensible to take the first observation as the time origin, either in AFT or proportional hazards type models. Instead, we advocate using age of the subject as the time scale. We account for the fact that baseline observations may be made at different ages in different patients via a two stage procedure. First, we marginally regress any potentially age-varying covariates against age, retaining the residuals. These residuals are then used as covariates in the fitting of either an AFT model or a proportional hazards model. We call the procedures residual accelerated failure time (RAFT) regression and residual proportional hazards (RPH) regression respectively. We compare standard AFT with RAFT, and demonstrate superior predictive ability of RAFT in real examples. In epidemiology, this has real implications in terms of risk communication to both patients and policy makers.
\end{abstract}


\section{Introduction}
Two common methods for analysing time to event data are proportional hazards models \citep[see, for example,][]{cox84} and accelerated failure time (AFT) models \citep[see, for example,][]{wei92}. When constructing these models, a suitable time scale must be chosen. For example, time may be measured as age of the subject, or on the other hand it may be measured as the subsequent time to event after an intervention or observation. In the latter case, age at study entry is often included as a covariate. If no risk-modifying event occurs at the time of first observation, this time is somewhat arbitrary \citep{korn97,thiebaut04}. Despite this, time-in-study is often used as the time scale in epidemiological studies.

The choice of time origin in survival analysis can be a difficult issue, but this choice is an essential part of any survival analysis \citep{cox84}. It can be seen as a question of aligning individuals under consideration in a sensible way by choosing `time zero'. This is first considered in \cite{farewell79}, who compare survival models for breast cancer incidence in women, taking the time origin as age, and alternatively as the time since birth of the first child. They find that taking time as age offers the most parsimonious solution. \cite{liestol02} also argue for the use of age or calendar time as the time origin when a time of incidence or initial observation is difficult to define, such as the situations we consider here. \cite{korn97} strongly advocate the use of age as the time origin in epidemiological studies, indeed they declare that using time since first observation as the time scale is `incorrect'.  

Using different time scales will lead to different estimates of parameters in the time to event model. \cite{thiebaut04} carry out an extensive simulation study on this issue, and also strongly advocate using age as the time origin rather than time on study.
On the other hand, \cite{fieberg09} advocate start of study as the most sensible time origin, following comparison of survival estimates using three different time origins: age, start of study and calendar time.

We do not discuss in this paper the related issue of having multiple scales in which time can be measured, for example age and mileage of a car \citep{oakes95}.

Suppose we have $n$ subjects. Denote by $s_i$ the age of individual $i$, which obviously changes with time. Suppose the individual is first observed at age $a_i$, at which point a series of covariates are extracted. Denote covariates that are not expected to vary with age by $\bm{x}_i$, and those that are expected to vary with age by $\bm{z}_{i,s_i}$, where the dependence on age is made explicit. The values of the time varying covariates that we actually observe are given by $\bm{z}_{i,a_i}$.
Let $t_i$ denote the time that has elapsed since individual $i$ was first observed.

This paper makes two contributions. First, we show formally why measuring time since observation, $t_i$, should be rejected in situations where no intervention is made at the time of observation. Second, we propose instead that the age of the individual, $s_i$, should be used as the time scale. As an extension to proportional hazards models, we propose a two-stage procedure called residual proportional hazards (RPH) regression, in order to correct for the fact that individuals' age-varying risk factors are observed at different ages $a_i$. Similarly, we propose the corresponding extension to the accelerated failure time model, which we call residual  accelerated failure time (RAFT) regression.

The paper is organised as follows. In Section \ref{sec:reject} we justify formally why we reject the idea of measuring time as time since first observation. In Section \ref{sec:method} we introduce RPH and RAFT regression. In Section \ref{sec:results} we demonstrate the performance of RAFT regression compared to existing methods of analysing time to event data. We conclude with a discussion in Section \ref{sec:discuss}.

\section{Measuring Time Since First Observation} \label{sec:reject}
In this Section we justify our stance that time should be measured as the age of the individual, $s_i$, rather than as time since first observation, $t_i$. This is a stance shared by many others \citep[see, for example,][]{breslow83,korn97,liestol02,thiebaut04}. We show, however, that using time-in-study as the time scale inevitably leads to elementary contradictions in situations where both the entry to study does not coincide with a risk modifying event, and age is expected to influence the time to the event of interest.

With $t$ denoting time to event, as measured from first observation,  let $m(t)$ denote the median time to event, or some other similar midpoint estimate of time to event (MTE). Let $s$ denote age of an individual. Consider comparing two individuals, $i$ and $j$, who are identical with respect to all their covariates except age. Without loss of generality assume that one individual is observed at age $s_i=a$, and the other is observed at age $s_j=0$. At least in situations where $t$ is the time after an observation that does not include an intervention, it seems reasonable to make the following two demands.

\begin{itemize}
\item The older person should have a shorter MTE:
\begin{equation} \label{eq:upper}
m(t|a) \leq m(t|0).
\end{equation}
When the event is naturally more common with advancing age, with all other things being equal, this is a reasonable assumption. Such is the case with common endpoints, such as death or disease incidence in epidemiological studies.

\item The fact that the older person has survived to age $a$ while avoiding the event should make his \emph{age} at the MTE greater than the age at the MTE of the younger person. This can be expressed as the inequality
\begin{equation} \label{eq:lower}
m(t|a)+a \geq m(t|0).
\end{equation} 
This requirement might be more subtle, but is logical in situations such as the following: suppose the event of interest is death from a chronic disease, and time is measured from a clinic appointment at recruitment to the study. Suppose we have two subjects, one older than the other, whose risk factors are otherwise identical. Then since the younger person is exposed to hazard between their current age, and the age of the older person, their expected age at death should be less.
\end{itemize}
Moreover, if we use time since first observation as the time scale, we should usually require that age should be included as a covariate in the model, since all other things being equal we would expect a failure event to occur to an older person sooner than a younger person.
We will now show that in a model where time is measured as time since first observation, it is not possible to both include age as a covariate, and satisfy the  two requirements specified by inequalities (\ref{eq:upper}) and (\ref{eq:lower}). We demonstrate this result for both AFT models and proportional hazards models.

\subsection{Proportional Hazards Model} \label{sec:cox}
Consider a proportional hazards model where age at observation, $a$, is the only covariate under consideration. This would be given by
$$
h(t;a) = \exp\big\{\beta f(a)\big\} h_0(t),
$$
for an individual of age $a$, for some $\beta$. Common choices for the age transformation function $f(\cdot)$ are $f(a) = a$ and $f(a) = \log(1+a)$ (where we add one to avoid the issue of taking logs of zero). Other covariates can be considered absorbed by the baseline hazard $h_0$. We are agnostic as to whether the baseline hazard is parametric or not. The survival function $S(t;a)$ is related to the hazard function via
$$
\frac{d}{dt} \log S(t;a) = -h(t;a),
$$
and hence the proportional hazards model implies that
$$
S(t;a) = S_0(t;a)^{\exp\{\beta f(a)\}},
$$
where $S_0$ denotes the baseline survival function.
A general form for the median of the proportional hazards model is therefore given by
\begin{equation} \label{eq:median-ph}
S^{-1}(0.5) = S_0^{-1}\big[0.5^{\exp\{-\beta f(a)\}}\big].
\end{equation}
Suppose that $f(a) = a$. See the remark at the end of this Section for $f(a) = \log(1+a)$.

\begin{claim} \label{cl:ph-byupper}
In order for the proportional hazards model to satisfy inequality (\ref{eq:upper}), we must have $\beta \geq 0$.
\end{claim}
\begin{proof}

Substituting (\ref{eq:median-ph}) into inequality (\ref{eq:upper}) gives
$$
S_0^{-1}(0.5) \geq S_0^{-1} \{0.5^{\exp(-\beta a)}\}.
$$
We consider the consequences of this inequality holding for the value of $\beta$. First we can apply the baseline survival function $S_0$ to both sides, as $S_0$ is a decreasing function, yielding
$$
0.5 \leq 0.5^{\exp(-\beta a)}.
$$
Taking logs of both sides and dividing through by $\log(0.5)$ (which is negative) yields
$$
\exp(-\beta a) \leq 1,
$$
from which we conclude that we require $\beta \geq 0$  to satisfy inequality (\ref{eq:upper}).
\end{proof}

\begin{claim} \label{cl:ph-bylower}
In order for the proportional hazards model to satisfy inequality (\ref{eq:lower}), we must have $\beta \leq 0$.
\end{claim}

\begin{proof}
Substituting (\ref{eq:median-ph}) into inequality (\ref{eq:lower}) gives
$$
S_0^{-1}(0.5) -  a \leq S_0^{-1} \big\{0.5^{\exp(-\beta a)}\big\}.
$$
Again, we consider the consequences of this inequality for $\beta$. Taking $S_0$ of both sides gives
$$
S_0\big\{S_0^{-1}(0.5) - a\big\}\geq 0.5^{\exp(-\beta a)}.
$$
Taking Taylor expansions of the left hand side about $S_0^{-1}(0.5)$ gives
$$
0.5 - \kappa a + O(a^2) \geq 0.5^{\exp(-\beta a)},
$$
where $\kappa = S_0'\big\{S_0^{-1}(0.5) \big\}$.
After standard manipulations, we then arrive at
\begin{equation} \label{eq:pre-hop}
\beta \leq - \frac{1}{a} \log \left[ \frac{\log\big\{0.5 - \kappa a + O(a^2) \big\}}{\log(0.5)} \right].
\end{equation}
We next note that inequality (\ref{eq:lower}) should hold for any reasonable $a$, and certainly for small values of $a$, so take limits as $a \downarrow 0$ of the right hand side of (\ref{eq:pre-hop}). This limit can be calculated using l'Hopital's rule,
$$
\lim_{a \downarrow 0} \frac{f(a)}{g(a)} = \frac{f'(0)}{g'(0)}.
$$
 taking
$$
f(a) = -\log \left[ \frac{\log\big\{0.5 - \kappa a + O(a^2) \big\}}{\log(0.5)} \right], \ \ g(a) = a.
$$
Applying this, we get
\begin{equation} \label{eq:pre-rescale}
\beta \leq \frac{2\kappa}{\log(0.5)}.
\end{equation}
Finally, recall that $\kappa = S_0'\big\{S_0^{-1}(0.5) \big\}$. It is a derivative of a survival function, and therefore must be non-positive. Since $\log(0.5)$ is negative, overall the right hand side of (\ref{eq:pre-rescale}) is non-negative.  As a derivative, $\kappa$ can be made arbitrarily small by rescaling the time axis. Such a time rescaling affects the baseline hazard, $h_0(t)$ only, and not the magnitude of $\beta$. Therefore we conclude
$\beta \leq 0$ to satisfy inequality (\ref{eq:lower}).
\end{proof}

Combining Claims \ref{cl:ph-byupper} and \ref{cl:ph-bylower} we must conclude that, in order to satisfy inequalities (\ref{eq:upper}) and (\ref{eq:lower}), we must have $\beta = 0$. In other words, having an age covariate is inconsistent with these requirements.

\begin{rk}
Setting $f(a) = \log(1+a)$ gives the same result. This is clear in Claim \ref{cl:ph-byupper}. For Claim \ref{cl:ph-bylower}, the argument would proceed identically up to (\ref{eq:pre-hop}). The only difference is that we have $g(a) = \log(1+a)$ rather than $g(a) = a$ for l'Hopital's rule, but these both yield the same $g'(0)$.
\end{rk}

\subsection{Accelerated Failure Time Model} \label{sec:aft}
An accelerated failure time (AFT) model is typically a log-linear model for the survival time. (Or, if not the survival time directly, the scale parameter, which is proportional to the survival time). Consider an AFT model where time is measured as time since first observation, and age, $a$, is the only covariate under consideration. This would be given by
\begin{equation} \label{eq:aft}
m(t) = \exp\big\{\alpha + \beta f(a)\big\}.
\end{equation}
Here, $\alpha$ is used to denote any other terms in the model, so may depend on background risk factors, etc. As before, the age transformation function $f(\cdot)$ is commonly the identity, $f(a) = a$; or a log transform, $f(a) = \log(1+a)$.
As before, we consider comparing two individuals, $i$ and $j$, one of age $s_i=a$, the other of age $s_j=0$, with all other covariates identical.

The following lemma will be useful.

\begin{lemma} \label{lem:timescale}
In the AFT model of Equation (\ref{eq:aft}), making an arbitrary change of factor $k$ to the time scale will alter the value of $\alpha$ but not $\beta$. Therefore make explicit this dependency by writing $\alpha(k)$. Then moreover, for any $\epsilon > 0$ there exists a scaling factor $k_\epsilon$ such that $\exp\big\{-\alpha(k_\epsilon)\big\} < \epsilon$.
\end{lemma}
\begin{proof}
Rescaling time by factor $k$ leads to 
$$
m(kt) = km(t) = k\exp\big\{\alpha + \beta f(a)\big\}.
$$
Setting $\alpha' = \log(k) + \alpha$ gives
$$
m(kt) = \exp\big\{\alpha' + \beta f(a)\big\},
$$
proving the first part of the lemma. For the second part, it suffices to take $k_\epsilon < \epsilon\exp(\alpha)$.
\end{proof}

As in the previous Section, we now consider the case where $f(a) = a$, and relegate comments about the case $f(a) = \log(1+a)$ to a remark at the end of this Section.

\begin{claim} \label{cl:aft-byupper}
In order for the accelerated failure time to satisfy inequality (\ref{eq:upper}), we must have $\beta \leq 0$.
\end{claim}

\begin{proof}
Substituting (\ref{eq:aft}) into inequality (\ref{eq:upper}) gives
$$
\exp(\alpha + \beta a) \leq \exp(\alpha),
$$
which immediately gives $\beta \leq 0$.
\end{proof}

\begin{claim} \label{cl:aft-bylower}
In order for the accelerated failure time to satisfy inequality (\ref{eq:lower}), we must have $\beta \geq 0$.
\end{claim}

\begin{proof}
Substituting (\ref{eq:aft}) into inequality (\ref{eq:lower}) gives
$$
\exp(\alpha)-a \leq \exp(\alpha + \beta a).
$$
Re-arranging,
$$
\beta \geq \frac{\log(1-ae^{-\alpha})}{a}.
$$
Now we require this for all $a>0$; taking limits as $a \downarrow 0$ yields
$$
\beta \geq -e^{-\alpha}.
$$
However, by Lemma \ref{lem:timescale}, for every $\beta<0$ there exists a time rescaling that gives an $\alpha$ violating this inequality, hence we must conclude
$\beta \geq 0$.
\end{proof}

Combining Claims \ref{cl:aft-byupper} and \ref{cl:aft-bylower} we must conclude that, again, in order to satisfy inequalities (\ref{eq:upper}) and (\ref{eq:lower}), we must have $\beta = 0$.

\begin{rk}
The same solution is reached for $f(a) = \log(1+a)$. This is trivial for Claim \ref{cl:aft-byupper}, and for Claim \ref{cl:aft-bylower} it follows because $\log(1+a) \approx a$ for small $a$.
\end{rk}

The problems of having two measurements of time in a survival model, both advancing at the same rate, are alluded to by \cite{liestol02}, who call this an `exchangeable' effect, i.e.\ there is an information redundancy between the two time measurements, $s_i$ and $t_i$.

\subsection{Example} \label{sec:eg}
The following example is based on the model of \cite{wilson08}. They construct a log-linear model for the scale parameter, $\lambda(x)$ say, of a Weibull density, which is proportional to the expected survival time, where the outcome is heart attack or stroke. Survival time is measured as time since an observation is made, and age at observation is included as a covariate.
Precisely, their model is
\begin{align*}
\lambda(x) = \exp\{ & 14.9756	-0.0159 \times \textrm{body mass index}\\ 
&-0.0571 \times \textrm{age in years} -0.4959 \times \textrm{smoker}\\
& -0.0070 \times \textrm{systolic blood pressure}\\ 
& -0.1432\times \textrm{total to HDL cholesterol ratio}\\
&	-0.3421 \times \textrm{diabetic} + 0.5139 \times \textrm{male}\}.
\end{align*}
Consider a reasonably healthy person, say body mass index = 20, a non-smoker, systolic blood pressure = 120, total to HDL cholesterol ratio = 3.5, non diabetic and male.
An example of  a counterintuitive result is that a 50 year old with the above risk factors and characteristics has
$$
P[\textrm{event over next 25 years}]=P[\textrm{event between age 50 and 75}]= 0.147.
$$ 
Yet, a 55 year old with the same characteristics has
$$
P[\textrm{event over next 20 years}]=P[\textrm{event between age 55 and 75}]= 0.159.
$$
This is clearly rather odd, as one could imagine this is the same person five years later, who remains in good health and has not had the event, yet their probability of an event by age 75 has increased.
This demonstrates when it is not sensible to allow inequality (\ref{eq:lower}) to be violated.

\section{RPH and RAFT Regression} \label{sec:method}
Many risk factors are expected to rise with advancing age. For example, blood pressure and cholesterol generally increase with age. This could be problematic for a time-to-event regression where age is time, since the levels of some risk factors may depend upon the age of the individual when they were observed, $a_i$. Not adjusting for this would give people observed at an older age a worse prognosis, simply because their risk factors were higher when measured. A younger person may get a better prognosis as his risk factors are lower when measured, despite the expected trajectory of those risk factors being higher than the older person. It therefore makes sense to correct for the age at which the risk factors have been measured.

We propose a two-stage procedure. In the first stage, risk factors that vary with time are regressed marginally with age, and the residuals from this regression retained. In the second stage, time to event regression is carried out (via either AFT or proportional hazards), but using the residuals of the time varying risk factors as covariates rather than their individual values.

\subsection{Stage 1}
Using the same notation as before, for each individual we have a pair $(a_i,\bm{z}_{i,a_i})$, consisting of the age at which the first observation was made, and the values of the age-varying risk factors at that age. For simplicity of exposition, suppose there is just one  age-varying risk factor, $z_{i,a_i}$. Suppose that the age-varying effect can be captured by a function $f$ that is common to all individuals $i=1,\ldots,n$:
\begin{equation} \label{eq:raft-resid}
z_{i,s_i} = f(s_i) + e_i,
\end{equation}
and of course, this is true in particular for $s_i = a_i$. We can therefore consider estimating $f$ by regressing $\{z_{i,a_i},i=1,\ldots,n\}$ on $\{a_i,i=1,\ldots,n\}$. For simplicity, in this paper we consider the case where $f$ is (or at least well approximated by) a linear function of age. After estimating $f$ we can recover the associated residuals $\{e_i, i=1,\ldots,n\}$. Assuming that $f$ correctly captures the dependence between the risk factor level and age, we can assume that these residuals are independent of age.

If we have multiple risk factors, $p$ say, we simply repeat the procedure outlined above, marginally for each risk factor, to obtain residuals for each age-varying risk factor, $\bm{e}_i = (e^1_i,\ldots, e^p_i)$, for each individual $i=1,\ldots,n$.

One could view this stage of the procedure as defining individuals by their percentiles in the age-specific risk factor distributions, which are likely to be more stable over age and time than the risk factors themselves.

\subsection{Stage 2}
In the second stage of the procedure, we  carry out the desired time to event regression. We suggest here proportional hazards regression or AFT regression, but any time to event regression methods would be valid.  Instead of using the age-varying risk factors as covariates, we use their residuals from Stage 1. Age is taken as the time scale. 

For AFT, the model is log-linear for a scale parameter $\lambda_i$, given by
\begin{equation} \label{eq:raft-aft}
\log(\lambda_i) = \beta_0 + \beta_x \bm{x}_i + \beta_z \bm{e}_i + \epsilon_i,
\end{equation}
where $\bm{x}_i$ is the vector of fixed risk factors, $\bm{e}_i$ are the residuals of the age-varying risk factors derived in stage 1, and the $\beta$s are the unknown parameters to be estimated.
Typically the scale parameter $\lambda_i$ is equal to, or at least in proportion to, the failure age $s_i$. There is also a shape parameter that is usually assumed to be fixed. Various (strictly non-negative) parametric baseline distributions are generally considered for AFT regression, such as gamma, Weibull, log-logistic, log-normal, Gompertz and extreme value. The distribution of the error term $\epsilon_i$ depends on the choice of baseline distribution.
In this case we call the entire procedure is called residual accelerated failure time (RAFT) regression. 

A proportional hazards model is log-linear in the hazard function of the \kth{i} individual $h_i$. The baseline hazard, $h_0$,  may or may not be specified parametrically. This model can be written as
\begin{equation} \label{eq:rph}
\log(h_i) = h_0 + \beta_x \bm{x}_i + \beta_z \bm{e}_i + \epsilon_i,
\end{equation}
where the terms are as for the AFT model.
In this case we call the entire procedure is called residual proportional hazards (RPH) regression. 

\section{Examples} \label{sec:results}
In this Section we compare the predictive ability of three different AFT-based regression methods:

\begin{enumerate}
\item Age as a covariate, time is measured as time from observation (denoted AFT-AC).
\item Age as time, age-varying covariates not adjusted for age (denoted AFT-NA).
\item Age as time, age-varying covariates adjusted for age by taking residuals from a regression on age (denoted RAFT).
\end{enumerate}

We focus on the comparison of AFT models: since AFT is a log-linear model for a parameter of the presumed distribution, they retain the original distributional form.  Proportional hazards models, on the other hand, only retain their original distribution if that distribution possesses the proportional hazards property. The three distributions we consider here are Weibull, log-normal and log-logistic. When a Weibull distribution is assumed, proportional hazards models and AFT models are equivalent.

We take as our data the US Health and Nutrition Examination Survey (NHANES-I) \citep{nhanes73}, and the outcome we are interested in is the time to the first non-fatal myocardial infarction (NFMI). Accelerated failure time regression models are fitted to the data, using the three baseline distributions, in each of the three paradigms discussed above. 

The predictive ability of the models is compared by splitting the cohort randomly in half, estimating the models on the first half (training data), and attempting to predict failure times in the second half (validation data).  
Specifically, we use the Brier score \citep{brier50} as the measure of predictive ability of the models.  Other measures of predictive ability of a survival model are considered in \cite{korn90}. We discretize time into years. Let $n$ be the number of individuals.  Let $s_i$ denote the time in years at which individual $i$ is first observed. For the AFT-AC model, $s_i = 1$ for all $i$, whilst for the other two models $s_i$ is the age at first observation, denoted $a_i$. Let $t_i = \min(c_i,M)$, where $c_i$ denotes the time in years at which individual $i$ is censored (with $c_i = \infty$ if censoring does not occur), and $M$ is the maximum time in years at which observations can be made. For the AFT-AC model, $M$ is the maximum follow-up time over the $i$ individuals, whilst for the other two models $M$ is the maximum age reached by an individual whilst still alive and not lost to follow-up. The outcome of interest is whether NFMI has occurred. Therefore, let 
$$
o_{ij} = \left\{\begin{array}{ll} 1 &\textrm{Individual $i$ has had NFMI by time $j$}\\
0 &\textrm{otherwise}. \end{array} \right.
$$
Let $g_{ij}$ be the predicted probability from the model that NFMI has occurred. This is given by
$$
g_{ij} = \left\{\begin{array}{ll} F(j) &\textrm{for AFT-AC}\\
\frac{F(j+a_i) - F(a_i)}{1-F(a_i)} &\textrm{for AFT-NA and RAFT,} \end{array}   \right.
$$
where $F$ is the cumulative distribution function of the assumed distribution.
Finally, let $N = \sum_{i=1}^n (t_i - s_i)$. Then we calculate the Brier score for predictive ability as
$$
B = \frac{1}{N}\sum_{i=1}^n \sum_{j=s_i}^{t_i} (g_{ij} - o_{ij})^2.
$$
A model with good predictive ability should produce a small Brier score.

We repeat this procedure 100 times, with random splits each time into the training and validation set, and report as the final Brier score the average of the 100 runs. The results are given in Table \ref{tab:predict-comp}.

[Table 1 about here]

Of the three methods, AFT-NA gives the largest Brier score in all situations. This is not surprising, since this method makes no correction for the age of the participant at observation. AFT-AC gives a moderate Brier score; this method corrects for age, but in a manner that has been shown, in Section \ref{sec:reject}, to be incoherent. RAFT always gives the lowest Brier score, so is the optimal model of the three for this data; this model corrects for age in a coherent manner. It is also noteworthy that there is stability in the Brier scores for RAFT across the three distributions that have been considered here.

\section{Discussion} \label{sec:discuss}
We have shown that measuring time from a baseline that consists of an observation of an individual, but no intervention, is not sensible. In such situations, age should always be chosen as the time scale. On the other hand, it can make sense to measure the time from a baseline that consists of an intervention, since in this case it is entirely reasonable that inequality (\ref{eq:lower}) is violated; for example, in randomised controlled trials. The modeller should, however, be aware that inequality (\ref{eq:lower}) is violated in this case, and think carefully about whether this is reasonable.

We propose new methods called residual proportional hazards (RPH) and residual accelerated failure time (RAFT) regression. These are two stage procedures: the first stage regresses age-varying risk factors on age to eliminate the effect of patients entering the study at different ages; the second stage carries out time to event regression (either AFT regression or proportional hazards regression), replacing the age-varying risk factors by their residuals derived from the first stage.

The additional predictive capability obtained from this two-stage procedure have implications for communicating risk to patients. Existing risk scores, commonly based on AFT-AC (AFT with age as a covariate, and the time origin as time of first observation) are shown in Section \ref{sec:results} to have weaker predictive ability, and in Section \ref{sec:reject} to lead to logical contradictions in the risk score assigned. Therefore, we believe that risk scores based on the methods outlined in this paper will lead to improved mortality estimates, which will enable better decisions to be made, both at the policy level and the individual patient level.

Epidemiologists often refer to three principal sets of factors: time; place; and person. We note that RAFT handles two dimensions of time in a way that is analogous to the multi-stage regression methods that the field of Geostatistics employs to handle two dimensions of space \citep{diggle07}. 

Extensions to RAFT regression can be considered. First, in this paper we have considered only linear forms for the function $f$ in Equation (\ref{eq:raft-resid}). Clearly this will often be unrealistic. Non-parametric estimates of this function could be considered, for example. In addition, one may wish to consider interactions between the age-risk factor relationship for multiple age-varying risk factors. This would involve replacing a series of marginal models of the form given in Equation (\ref{eq:raft-resid}) by a joint model. 

Second, the estimation of $f$ is prone to a `survivor bias'. Subjects with higher values of risk factors tend to die younger, hence subjects recruited at an older age have probably had lower risk factor trajectories, on average, than those recruited at a younger age. One could correct for this by constructing a procedure that iterates between Stage 1 and Stage 2 of the method outlined in this paper. An estimate of the risk factor effect on survival can be estimated from Stage 2, which can then be passed back to Stage 1 to account for this bias in estimating $f$; then the risk factor effects re-estimated in Stage 2 based on the new residuals, and so on. Note that this would also require the risk factor effect on deaths from other causes (causes not of interest) and other censoring mechanisms to be estimated in the second stage.

Finally, we have not considered the propogation of uncertainty from the first stage of the RAFT regression procedure to the second, which may lead to standard errors being underestimated.

\section{Funding}
This work was funded by the National Institute for Health Research grant Greater Manchester Collaboration for Applied Health Research and Care, in collaboration with the Medical Research Council IMPACT project.

\section{Acknowledgements}
The authors would like to thank David Spigelhalter for useful comments.

\bibliographystyle{jtbnew}
\bibliography{bib-raft}

\begin{table}[p]
\caption{Mean Brier scores \label{tab:predict-comp}}
\begin{center}
\begin{tabular}{cccc}
Baseline &AFT-AC &AFT-NA &RAFT\\
\hline
Weibull &0.0217 &0.0605 &0.0149\\
Log-Normal &0.0221 &0.0384 &0.0148\\
Log-Logistic &0.0160 &0.0180 &0.0148\\
\end{tabular}
\end{center}
\end{table}

\end{document}